\documentclass{amsart}

\usepackage{amsmath}
\usepackage{amssymb}
\usepackage{amsthm}
\usepackage{amsaddr}

\usepackage{graphicx}

\title[The discrete Toda equation revisited]{
The discrete Toda equation revisited\\
dual $\beta$-Grothendieck polynomial, ultradiscretization, and static solitons}
\author{Shinsuke Iwao and Hidetomo Nagai}
\address{Department of Mathematics, Tokai University, 4-1-1, Kitakaname, Hiratsuka, Kanagawa 259-1292, Japan.}
\email{iwao@tokai.ac.jp}
\email{hdnagai@tokai-u.jp}

\date{\today}

\newtheorem{thm}{Theorem}[section]

\newtheorem{prop}[thm]{Proposition}

\newtheorem{lemma}[thm]{Lemma}
\newtheorem{rem}[thm]{Remark}

\newcommand\ZZ{\ensuremath{\mathbb{Z}}}
\newcommand\CC{\ensuremath{\mathbb{C}}}
\newcommand\RR{\ensuremath{\mathbb{R}}}

\def\zet#1{\left\vert {#1} \right\vert}

\newcounter{xax}%
\newcounter{yax}%

\usepackage{mathtools}

\makeatletter

\def\mWord#1{\text{\ $\mWord@A{#1}$}} 
\def\mWord@A#1{%
\bgroup
\let\\=\crcr
\def\,{& &}
\def\u##1{\underline{##1}}
\def\d##1{\dot{##1}}
  \vcenter{\offinterlineskip
      \halign{
       &\vrule width 0pt height 8pt depth 3pt {\hfil${##}$\hfil} \!\! & \!\!{##},\! \cr #1 \crcr
      }
  }
\egroup
}

\makeatother

\newcommand{\Array}[1]{\left(\!\mWord{#1}\!\right)}
\def\V#1{\pmb{#1}}

\newcommand{\ep}{\varepsilon}

\begin{document}
\maketitle

\begin{abstract}
This paper presents a study of the discrete Toda equation
\[
(\tau_n^t)^2+\tau_{n-1}^t\tau_{n+1}^t=\tau_n^{t-1}\tau_n^{t+1},
\]
that was introduced in 1977. In this paper, it has been proved that the algebraic solution of the discrete Toda equation, obtained via the Lax formalism, is naturally related to the dual Grothendieck polynomial, which is a $K$-theoretic generalization of the Schur polynomial.
A tropical permanent solution to the ultradiscrete Toda equation has also been derived. The proposed method gives a tropical algebraic representation of static solitons. Lastly, a new cellular automaton realization of the ultradiscrete Toda equation has been proposed.
\end{abstract}

\section{Introduction}

\subsection{Algebraic solution to the discrete Toda equation}
The Toda equation
\[
  \frac{d}{dt^2}\log(1+V_n(t))=V_{n+1}(t)-2V_n(t)+V_{n-1}(t)
\]
was proposed as a model equation of motion in a one-dimensional lattice of particles with the nearest neighbor interaction~\cite{toda1967wave}.  Today the equation is known as a good example of an integrable equation owing to its rich structures.  
In this paper, we intend to study time-discretization of the Toda equation\footnote{There exist at least two 'discrete Toda equations' which are famous and well-investigated (as far as we know).
See Remark \ref{rem:no.1}.}, which was originally given by Hirota~\cite{hirota1977}:
\begin{equation}\label{eq:originalform}
\frac{u_n^{t-1}u_n^{t+1}}{(u_n^{t})^2}=\frac{(1-\delta ^2+\delta^2u_{n-1}^{t})(1-\delta ^2+\delta^2u_{n+1}^{t})}{(1-\delta ^2+\delta^2u_n^{t})^2}.
\end{equation}
The above equation (\ref{eq:originalform}) boils down to the following bilinear form.
\begin{equation}\label{eq:bilinear}
(\tau_n^t)^2+\tau_{n-1}^t\tau_{n+1}^t=\tau_n^{t-1}\tau_n^{t+1}
\end{equation}
through the variable transformation given by
\[
u_n^t=
\frac{1-\delta^2}{\delta^2}\frac{\tau_{n}^{t+1}\tau_{n+2}^{t+1}}{(\tau_{n+1}^{t+1})^2}.
\]
Moreover, if we define two new variables given by
\begin{equation}\label{eq:a-and-b}
a_n^t:=\frac{\tau_n^t\tau_{n+1}^{t+1}}{\tau_n^{t+1}\tau_{n+1}^{t}},\qquad
b_n^t:=\frac{\tau_n^t\tau_{n+2}^{t+1}}{\tau_{n+1}^t\tau_{n+1}^{t+1}},
\end{equation}
the equation (\ref{eq:bilinear}) can be rewritten as 
\begin{equation}\label{eq:local}
a_n^{t+1}+b_{n-1}^{t+1}=a_n^t+b_n^t,\qquad
a_{n+1}^{t+1}b_n^t=a_n^tb_n^{t+1}.
\end{equation}
One can recover (\ref{eq:originalform}) from (\ref{eq:local}) by using the following relation.
\begin{equation}\label{eq:u}
u_n^t=\frac{1-\delta^2}{\delta^2}
\frac{b_{n}^{t}}{a_{n}^{t}}.
\end{equation}

We consider the simultaneous equations (\ref{eq:local}) with the boundary condition 
\begin{equation}\label{eq:boundarycondition}
a_0^t=a_{N+1}^t=1,\qquad b_0^t=b_{N}^t=0
\end{equation}
for some integer $N>0$.
Let
\[
X^t=
\left(
\begin{array}{cccc}
a^t_1 & 1 &  &  \\ 
 & a^t_2 & \ddots &  \\ 
 &  & \ddots & 1 \\ 
 &  &  & a^t_N
\end{array} 
\right),\qquad
Y^t=
\left(
\begin{array}{cccc}
1 & &  &  \\ 
-b^t_1 & 1 &  &  \\ 
 & \ddots & \ddots &  \\ 
 &  & -b_{N-1}^t & 1
\end{array} 
\right)
\]
be $N\times N$ matrices.
The simultaneous equation (\ref{eq:local}), therefore, admits the use of the discrete Lax formulation given as follows:
\begin{equation}\label{eq:Lax}
X^{t+1}Y^t=Y^{t+1}X^t.
\end{equation}
We define $L^t:=(Y^t)^{-1}X^t$.
Thus, (\ref{eq:Lax}) could be written as
\begin{equation}\label{eq:LaxMatrixForm}
Y^tL^{t}=L^{t+1}Y^t\qquad \mbox{or}\qquad X^tL^t=L^{t+1}X^t.
\end{equation}
Similar to other classical integrable systems (see, for example, \cite{hirota1980direct, hirota2004direct, KOSTANT1979195, sklyanin2013bispectrality}), one can construct algebraic solutions for arbitrary initial values (values at time $t=0$) of the discrete Toda equation via the Lax representation (\ref{eq:LaxMatrixForm}).

\begin{rem}\label{rem:no.1}
There exist at least two famous 'discrete Toda equations'.
One is the equation $(\ref{eq:originalform})$, which we investigate in this study, and the other is the one expressed by the following bilinear form
\begin{equation}\label{eq:otherform}
(\tau_n^t)^2+\tau_{n-1}^{t+1}\tau_{n+1}^{t-1}
=\tau_n^{t+1}\tau_n^{t-1}.
\end{equation}
(Compare with $(\ref{eq:bilinear})$.)
\end{rem}

\subsection{Algebraic solution and dual Grothendieck polynomials}

Obtaining an algebraic solution to the discrete Toda equation, through use of the Lax formulation (\ref{eq:LaxMatrixForm}), with an arbitrary (generic) initial condition is relatively straightforward. 
(See \S \ref{sec:2.1}--\S \ref{sec:2.3} for details.)
If all eigenvalues of the Lax matrix $L^t$ degenerate to one value (as (\ref{eq:degenetate}) in \S \ref{sec:2.4}), there exists a natural algebraic relation between these solutions and the dual Grothendieck polynomials, which are essentially $K$-theoretic analogues of Schur polynomials~\cite{lam2007combinatorial}.
Interestingly, 'the other' discrete Toda equation (\ref{eq:otherform}) corresponds to usual Schur polynomials (Remark \ref{rem:comparement}).

Recently, several researchers have reported interesting relations between '$K$-theoretic' objects and classical integrable systems. Motegi and Sakai~\cite{motegi2013vertex} discovered a remarkable relation between Grothendieck polynomials and algebraic solutions to certain integrable systems (TASEP). Very recently, Ikeda, Iwao, and Maeno \cite{ikeda2017peterson} constructed a ring isomorphism between the quantum $K$-theory of the complex flag variety $Fl_n$ and the $K$-theory of the affine Grassmannian $Gr_{SL_n}$ by using the mechanics of the relativistic Toda equation.
From this, it may be inferred that the discrete Toda equation $(\ref{eq:originalform})$ is amenable to {\it Grothendieck polynomial-type solutions}, while 'the other' discrete Toda equation $(\ref{eq:otherform})$ is amenable to {\it Schur polynomial-type solutions}.

\subsection{Ultradiscretization}
Another topic discussed in this paper is ultradiscretization.
Let us introduce the transformations $u_n^t=e^{\frac{U_n^t}{\varepsilon }}$ and $\delta =e^{-\frac{L}{2\varepsilon }}$ with a parameter $\ep>0$ and a positive constant $L>0$ for (\ref{eq:originalform}).
Then, by applying $\ep \log$ to both sides and taking the limit $\ep\to 0^+$, we obtain the ultradiscrete Toda equation~\cite{matsukidaira1997toda} given by
\begin{equation}  \label{eq:uToda}
  U_{n}^{t+1}-2U_{n}^{t}+U_{n}^{t-1} = \max[ 0, U_{n+1}^{t}-L]-2\max[ 0, U_{n}^{t}-L]+\max[ 0, U_{n-1}^{t}-L].  
\end{equation}  
In \cite{matsukidaira1997toda}, Matsukidaira {\it et al.}\ derived an soliton solution to (\ref{eq:uToda}) by ultradiscretizing the soliton solution {to} the discrete Toda equation (\ref{eq:originalform}).  
Their solution is expressed by 
\begin{equation}  \label{Nsolitonsolution}
\begin{aligned}
  &U^t_n = T^t_{n+1}-2T^t_n+T^t_{n-1}, \\
  &T^t_n = \max_{\mu _j\in \{0, 1 \}}\left[ \sum _{j=1}^N \mu _j S_j(t, n) -\sum _{1\le i<j\le N} \mu _i \mu_j (P_j+\sigma _i\sigma _j Q_j)\right], \\
  &S_j(t, n) = P_jn -\sigma _j Q_j t +C_j, \\
  &0\le P_1\le P_2\le \dots \le P_N, \quad Q_j= \max[0, P_j-L ], \quad \sigma _j\in \{-1, 1 \}.    
\end{aligned}
\end{equation}
Here $N$, $L$ are positive integers, and $P_j$, $C_j$ are arbitrary parameters.  
The operator $\max_{\mu _j\in \{0, 1 \}}f(\mu _1, \mu _2, \dots, \mu _N)$ denotes the maximum value amongst $2^N$ possible values of $f(\mu _1, \mu _2, \dots, \mu _N)$ obtained by replacing each $\mu _j$ by $0$ or $1$.  
It is known that the solution $U^t_n$ defined by (\ref{Nsolitonsolution}) possesses properties of solitary waves and soliton interactions~\cite{matsukidaira1997toda}.  
It can be verified that $U_n^t$ always takes non-negative values.  On the other hand, Hirota proposed another solution to (\ref{eq:uToda}), which is called 'the static-soliton'~\cite{hirota2009new}, and is expressed by
\begin{equation}\label{eq:Hirotasolution}
U^t_n = T^t_{n+1}-2T^t_n+T^t_{n-1}, \quad
T^t_n = C \sum _{j=j_0}^{j_1} \min[0, n-j],
\end{equation}
where $C$ is a positive parameter and $j_0\le j_1$ are integers.  One can verify that now $U_n^t$ (\ref{eq:Hirotasolution}) may take negative values, which implies that the equation (\ref{eq:uToda}) should be amenable to different solutions.
It should be natural to expect the existence of solutions to the discrete Toda equation (\ref{eq:local}) whose ultradiscretization yields static-solitons.
(Here we would like to note that Hirota~\cite{hirota2009new} showed there is no time independent solution, to the discrete Toda equation, other than the trivial solution.)

In this paper, we aim to provide an answer to the question by constructing {\it tropical permanent solutions} for arbitrary initial values of the ultradiscrete Toda equation.
More precisely, we intend to show that the tropical tau function $T_n^t$ (\ref{eq:tropicalT}) can be used to solve (\ref{eq:uToda}) by setting 
\[
U_n^t=A_n^t-B_n^t+L,\quad 
A_n^t=T^t_n+ T^{{t+1}}_{n+1}-T^{t+1}_{n}-T ^t_{n+1}, \quad  B_n^t=T^t_n+T^{{t+1}}_{n+2}-T^t_{n+1}- T^{{t+1}}_{n+1}.
\]
See \S \ref{sec:3.2} for details.
The new parameters $A_n^t,B_n^t$ represent the ultradiscretization of $a_n^t,b_n^t$, which satisfy the evolution equation (\ref{eq:udToda}).
A new cellular automaton realization of the system $\{A_n^t,B_n^t\}$ is proposed in \S \ref{sec:3.3}.

\subsection{Organization of the paper}

In Section \ref{sec:2}, we provide an algebraic solution to the discrete Toda equation (\ref{eq:local}) with the boundary condition (\ref{eq:boundarycondition}) using the Lax formulation.
Although this sequence of calculations is an established practice, we, nonetheless, have provided its details in \S \ref{sec:2.1}--\S \ref{sec:2.3} in order to make this paper self-contained.
In \S \ref{sec:2.4}, dual Grothendieck polynomials, as special solutions to the discrete Toda equation (\ref{eq:originalform}), have been derived.

In Section \ref{sec:3}, the tropical permanent solution to the ultradiscrete Toda equation, obtained by ultradiscretizing the algebraic solution defined in the previous section, has been provided.
The solution realizes the behavior of the solution given in \cite{matsukidaira1997toda,hirota2009new}.
A new cellular automaton realization of the system is proposed in \S \ref{sec:3.3}.
A concrete example and concluding remarks are contained in Section \ref{sec:4}.

\section{Solution to the discrete Toda equation}\label{sec:2}

\subsection{Lax formulation and the spectrum problem}\label{sec:2.1}

An algebraic solution to the discrete Toda equation via Lax formulation (\ref{eq:LaxMatrixForm}) is derived as follows.
Let 
\[
f(\lambda):=\det(\lambda E_N-L^t)=\lambda^N-I_1\lambda^{N-1}+\cdots+(-1)^NI_N
\]
be the characteristic polynomial of $L^t$, which is $t$-invariant due to (\ref{eq:LaxMatrixForm}).
Define the $\CC$-algebra $\mathcal{O}=\CC[\lambda]/(f(\lambda))$, which is $N$ dimensional as a $\CC$-vector space.
At the same time, the spectral problem
\begin{equation}\label{eq:spectrum}
L^t\V{v}^t=\lambda \V{v}^t,\qquad (\V{v}^t\in \CC^N),
\end{equation}
which is equivalent to 
\begin{equation}\label{eq:spectrum2}
(\lambda Y^t-X^t)\V{v}^t=\V{0},
\end{equation} has been considered.
We denote $(i,j)^{\mathrm{th}}$ minor of the matrix $\lambda Y^t-X^t$ by $M_{i,j}$ and define
\[
\Delta_{i,j}:=(-1)^{i+j}\det M_{i,j}.
\]
Therefore, the vectors 
\[
\V{p}:=(\Delta_{N,1},\Delta_{N,2},\dots,\Delta_{N,N})^T,\qquad
\V{q}:=(\Delta_{1,1},\Delta_{1,2},\dots,\Delta_{1,N})^T
\]
satisfy the following properties.
\begin{enumerate}
\item Property A
\begin{itemize}
\item The $i^{\mathrm{th}}$ entry of $\V{p}$, which is denoted by $\V{p}_i$, is a monic polynomial of degree $(i-1)$ in $\lambda$.
\item The $i^{\mathrm{th}}$ entry of $\V{q}$, which is denoted by $\V{q}_i$, is of the form $\lambda^{N-1}\times (\mbox{a polynomial of degree $(N-i)$ in $\lambda^{-1}$})$.
\item If $\lambda$ is a root of $f(\lambda)$, both $\V{p}$ and $\V{q}$ are eigenvectors of $L^t$.
In other words, $\V{p}$ and $\V{q}$ are solutions of the spectral problem (\ref{eq:spectrum}).
\end{itemize}
\item Property B
\begin{itemize}
\item The constant term of $\V{p}_i$ is $(-1)^{i-1}a_1^ta_2^t\cdots a_{i-1}^t$.
\item The coefficient of $\lambda^{N-1}$ of $\V{q}_i$ is $b_1^tb_2^t\cdots b_{i-1}^t$．
\end{itemize}
\end{enumerate}

Thus, it can be directly proved that as a matrix over $\mathcal{O}$, the co-rank of $\lambda Y^t-X^t$ is $1$.
Hence, the eigenvector of $L$ must be unique up to a constant multiple.
In terms of $\mathcal{O}$, we have the following lemma.
\begin{lemma}
Let $\CC[\lambda]\to \mathcal{O}$ be the natural surjection and $\overline{\V{p}},\overline{\V{q}}\in \mathcal{O}^N$ be the image of $\V{p},\V{q}\in \CC[\lambda]^N$.
Thus, there exists some $F\in \mathcal{O}$ such that
\begin{equation}\label{eq:inverseproblem}
\overline{\V{q}}=-F\cdot \overline{\V{p}}.
\end{equation}
$($The minus sign is used for the convenience of the calculation below$)$
\end{lemma}

Now, let us consider the inverse problem.
\begin{center}
for given $F\in \mathcal{O}$, recover $\V{p}$ and $\V{q}$ with (\ref{eq:inverseproblem}).
\end{center}
The answer is as follows.
Equation (\ref{eq:inverseproblem}) is rewritten as 
\begin{equation}\label{eq:linearproblem}
\left(
\begin{array}{cccc@{\ \ \vrule\ \ }cccc}
\ast & \ast & \cdots & \beta_1 & 1 &  &  &  \\ 
 & \ast & \cdots & \beta_2 & \alpha_2 & 1 &  &  \\ 
 &  &   \ddots & \vdots & \vdots & \vdots & \ddots   \\ 
 &  &  &   \beta_N& \alpha_N & \ast & \cdots & 1
\end{array} 
\right)
\Array{1\\\lambda\\ \vdots\\ \lambda^{N-1}\\ F\\F\lambda\\\vdots\\F\lambda^{N-1}}
\equiv \V{0}\qquad \mod f(\lambda).
\end{equation}
Each $\ast$ is a coefficient of $\V{p}_i$ or $\V{q}_i$ as a polynomial in $\lambda$.
Therefore, we have
\begin{equation}\label{eq:atoalpha}
\alpha_i=(-1)^{i-1}a_1^ta_2^t\cdots a_{i-1}^t,\quad
\beta_i=b_1^tb_2^t\dots b_{i-1}^t.
\end{equation}
By applying Cramer's rule to the matrix equation (\ref{eq:linearproblem}), one can express the entries of the matrix to be ratio of determinants.
Let 
$$
\V{c}:\mathcal{O}\to \CC^N
$$
be an arbitrary linear isomorphism.
We, therefore, have
\begin{equation}\label{eq:alpha}
\alpha_i=(-1)^{i-1}\frac{
\zet{\V{c}(\lambda^{i-1}),\V{c}(\lambda^{i}),\dots,\V{c}(\lambda^{N-1}),
\V{c}(F\lambda),\V{c}(F\lambda^{2}),\dots,\V{c}(F\lambda^{i-1})
}
}
{
\zet{\V{c}(\lambda^{i-1}),\V{c}(\lambda^{i}),\dots,\V{c}(\lambda^{N-1}),
\V{c}(F),\V{c}(F\lambda),\dots,\V{c}(F\lambda^{i-2})
}
}
\end{equation}
and 
\begin{equation}\label{eq:beta}
\beta_i=(-1)^i
\frac{
\zet{\V{c}(\lambda^{i-1}),\V{c}(\lambda^{i}),\dots,\V{c}(\lambda^{N-2}),
\V{c}(F),\V{c}(F\lambda),\dots,\V{c}(F\lambda^{i-1})
}
}
{
\zet{\V{c}(\lambda^{i-1}),\V{c}(\lambda^{i}),\dots,\V{c}(\lambda^{N-1}),
\V{c}(F),\V{c}(F\lambda),\dots,\V{c}(F\lambda^{i-2})
}
}.
\end{equation}
Note that they are invariant under any exchange of $\V{c}$.
As long as the denominators are not $0$, they recover values of $a_n^t$ and $b_n^t$.

\begin{rem}\label{rem:propor}
The expressions $(\ref{eq:alpha},\ref{eq:beta})$ are invariant also under the transform $F\mapsto cF$ $(c\in \CC^\times)$.
Thus, we may assume $F$ to be an element of $\mathcal{O}/\CC^\times$ without any loss of generality.
\end{rem}

In the remaining part of this section, we denote $F=F^t$, $\V{p}=\V{p}^t$, and $\V{q}=\V{q}^t$, {\it etc.}\ to emphasize the $t$-dependencies of these quantities.
From the discrete Lax equation (\ref{eq:LaxMatrixForm}) and (\ref{eq:spectrum}), we have a spectral problem at time $t+1$ given by 
\begin{equation}\label{eq:timeevolution}
L^{t+1}(X^t\V{v}^t)=\lambda (X^t\V{v}^t),\qquad
L^{t+1}(Y^t\V{v}^t)=\lambda (Y^t\V{v}^t).
\end{equation}
Let $\V{p}':=Y^t\V{p}^t$ and $\V{q}'=X^t\V{q}^t$.
Due to the shapes of the matrices $X^t$, $Y^t$ and Property A, the pair $(\V{p}',\V{q}')$ also satisfies Property A.
Because any pair of vectors with Property A uniquely restores the Lax matrix $L$ with Property B, we have 
\[
\V{p}^{t+1}=Y^t\V{p}^t,\qquad \V{q}^{t+1}=\theta X^t\V{q}^t,\quad (\exists \theta\in \CC^\times).
\]
Substituting them into (\ref{eq:inverseproblem}), we obtain
\[
\theta\overline{X^t\V{q}^t}=-F^{t+1}\overline{Y^t\V{p}^t},
\]
which implies
\[
\theta \lambda \overline{\V{q}^t}=\theta \overline{(Y^t)^{-1}X^t\V{q}^t}=-F^{t+1}\overline{\V{p}^t}.
\]
Comparing with (\ref{eq:inverseproblem}), we finally derive
\begin{equation}\label{eq:timeevoofF}
F^{t+1}=\lambda\theta \cdot F^t.
\end{equation}
Although this equation contains the unknown constant $\theta$，the expression (\ref{eq:timeevoofF}) still determines the time evolution of $L^t$ without any ambiguity.
See Remark \ref{rem:propor}.

\subsection{The determinantal formula for the tau function}

From (\ref{eq:timeevoofF}) and the remark \ref{rem:propor}, one may identify $F=F^t$ with $\lambda^tF^0$ without any problem.
Hereafter, we assume any root of $f(\lambda)$ to be non-zero.
Let $M_\lambda:\CC^N\to \CC^N$ represent the $\CC$-linear map$\colon$
\[
\CC^N\stackrel{\V{c}^{-1}}{\longrightarrow}\mathcal{O}\stackrel{\times \lambda}{\longrightarrow} \mathcal{O}\stackrel{\V{c}}{\longrightarrow}\CC^N.
\]
By assumption, this is invertible.
Let $D=\det{M_\lambda}\neq 0$.
The numerator of (\ref{eq:alpha}) can be rewritten as 
\begin{align*}
&\zet{\V{c}(\lambda^{i-1}),\V{c}(\lambda^{i}),\dots,\V{c}(\lambda^{N-1}),
\V{c}(F\lambda),\V{c}(F\lambda^{2}),\dots,\V{c}(F\lambda^{i-1})
}\\
&=D\cdot 
\zet{\V{c}(\lambda^{i-2}),\V{c}(\lambda^{i-1}),\dots,\V{c}(\lambda^{N-2}),
\V{c}(F),\V{c}(F\lambda),\dots,\V{c}(F\lambda^{i-2})
}.
\end{align*}
By putting
\begin{equation}\label{eq:tau}
\tau^t_n:
=
\zet{\V{c}(\lambda^{{n-1}}),\V{c}(\lambda^{n}),\dots,\V{c}(\lambda^{N-1}),
\V{c}(F^0\lambda^t),\V{c}(F^0\lambda^{t+1}),\dots,\V{c}(F^0\lambda^{t+n-2})
},
\end{equation}
we get
\begin{equation}\label{eq:modfy-a-and-b}
a_n^t=\frac{\tau^t_n\tau^{{t+1}}_{n+1}}{\tau^{t+1}_{n}\tau^t_{n+1}},\qquad
b_n^t=-\frac{\tau^t_n\tau^{{t+1}}_{n+2}}{\tau^t_{n+1}\tau^{{t+1}}_{n+1}}
\end{equation}
from (\ref{eq:atoalpha}).

\subsection{Double Casorati determinant}\label{sec:2.3}

By choosing a specific isomorphism $\V{c}:\mathcal{O}\to \CC^N$, one can derive an explicit formula for the tau function.
One typical example is the double Casorati determinant formula, which has been given below.

Assume that all eigenvalues of $f(\lambda)$ are distinct. Thus,
\[
f(\lambda)=\textstyle\prod_{i=1}^{N}(\lambda-\lambda_i),\qquad (i\neq j\Rightarrow\lambda_i\neq \lambda_j).
\]
By the Chinese remainder theorem, the following map represents linear isomorphism.
\begin{equation}\label{eq:linear_isom_c}
\V{c}:\mathcal{O}\to \CC^N;\qquad \varphi(\lambda)\mod{f(\lambda)}\mapsto (\varphi(\lambda_1),\dots,\varphi(\lambda_N)).
\end{equation}
For this $\V{c}$, the tau function $\tau_n^t$ can be expressed as
\begin{equation}\label{eq:explicit-form-of-tau}
\tau_n^t=
\left|
\begin{array}{cccc@{\ \vline \ }ccccc}
\lambda_1^{n-1} & \lambda_1^n & \cdots & \lambda_1^{N-1} &  f_1\lambda_1^t & f_1\lambda_1^{t+1} & \cdots & f_1\lambda_1^{t+n-2}   \\ 
\lambda_2^{n-1} & \lambda_2^n & \cdots & \lambda_2^{N-1} &  f_2\lambda_2^t & f_2\lambda_2^{t+1} & \cdots & f_2\lambda_2^{t+n-2}   \\ 
\vdots & \vdots & \cdots & \vdots & \vdots & \vdots & \cdots &\vdots \\ 
\lambda_N^{n-1} & \lambda_N^n & \cdots & \lambda_N^{N-1} &  f_N\lambda_N^t & f_N\lambda_N^{t+1} & \cdots & f_N\lambda_N^{t+n-2}   
\end{array} 
\right|,
\end{equation}
where $\V{c}(F^0)=(f_1,\dots,f_N)^T$.
By applying the Laplace expansion along the columns between $1^{\mathrm{st}}$ and $(N-n+1)^{\mathrm{th}}$ position, we have
\begin{align}  \label{eq:expansionoftau}
\tau_n^t=
\sum_{\genfrac{}{}{0pt}{1}{[N]=\V{j}\sqcup\V{i}}{\sharp\V{i}=n-1}}
(-1)^{\ep}
\prod_{j\in \V{j}}\lambda_j^{n-1} \cdot
\prod_{i\in \V{i}}f_i\lambda_i^{t}\cdot
\prod_{\{j_1<j_2\}\subset \V{j}}(\lambda_{j_1}-\lambda_{j_2})\cdot
\prod_{\{i_1<i_2\}\subset \V{i}}(\lambda_{i_1}-\lambda_{i_2}),
\end{align}
where $\ep=i_1+\dots+i_{n-1}+\frac{(2N-n+2)(n-1)}{2}$.  

\subsection{The dual stable Grothendieck polynomial}\label{sec:2.4}

In some special cases, the tau function is naturally related to the dual stable Grothendieck polynomial.
Let us consider the most degenerate case, wherein
\begin{equation}\label{eq:degenetate}
f(\lambda)=(\lambda-\gamma)^N,\qquad (\gamma\neq 0).
\end{equation}
(This assumption is valid in \S \ref{sec:2.4} only.)
The $\CC$-algebra $\mathcal{O}$, therefore is expressed as 
\[
\mathcal{O}=\CC[\lambda]/((\lambda-\gamma)^N)=\CC[\mu]/(\mu^N),
\]
where $\mu:=\lambda-\gamma$.
The $\CC$-valued functions $c_0,\dots,c_{N-1}$ over $\mathcal{O}$ could be defined as
\[
c_i(\alpha_0+\alpha_1\mu+\alpha_2\mu^2+\cdots+\alpha_{N-1}\mu^{N-1}):=\alpha_i.
\]
Let $\CC[\mathcal{O}]$ be the polynomial ring over $\mathcal{O}\colon\CC[\mathcal{O}]=\CC[c_0,\dots,c_{N-1}]$.

Next, we set $\beta:=\gamma^{-1}$. 
Let
\[
\tau_n=\zet{\V{c}(\lambda^{{n-1}}),\V{c}(\lambda^{n}),\dots,\V{c}(\lambda^{N-1}),
\V{c}(F\lambda),\V{c}(F\lambda^2),\dots,\V{c}(F\lambda^{n-1})
}
\]
be the $n^{\mathrm{th}}$ tau function corresponding to $F\in \mathcal{O}$, which can be rewritten as follows.
\begin{align*}
\tau_n
&=D^{n-1}
\zet
{
\V{c}(1),\V{c}(\lambda),\dots,\V{c}(\lambda^{N-n}),
\V{c}(F\lambda^{-(n-2)}),\dots,\V{c}(F\lambda^{-1}),\V{c}(F)
}\\
&=D^{n-1}
\zet
{
\V{c}(1),\V{c}(\lambda),\dots,\V{c}(\lambda^{N-n}),
\V{c}(F(\tfrac{\gamma-\lambda}{\gamma\lambda})^{n-2}),\dots,\V{c}(F\tfrac{\gamma-\lambda}{\gamma\lambda}),\V{c}(F)
}\\
&=D^{n-1}
\zet
{
\V{c}(1),\V{c}(\mu+\gamma),\dots,\V{c}((\mu+\gamma)^{N-n}),
\V{c}(F(\tfrac{-\mu}{\gamma(\mu+\gamma)})^{n-2}),\dots,\V{c}(F\tfrac{-\mu}{\gamma(\mu+\gamma)}),\V{c}(F)
}\\
&=D^{n-1}\beta^{\frac{(n-1)(n-2)}{2}}
\zet
{
\V{c}(1),\V{c}(\mu),\dots,\V{c}(\mu^{N-n}),
\V{c}(F),\V{c}(F\tfrac{\mu}{\mu+\gamma}),\dots,\V{c}(F(\tfrac{\mu}{\mu+\gamma})^{n-2})
}.
\end{align*} 
We fix the linear isomorphism $\V{c}:\mathcal{O}\to \CC^N$ as
\[
\V{c}:=c_0\V{e}_1+c_1\V{e}_2+\dots+c_{N-1}\V{e}_N,\qquad\quad \V{e}_i=(0,\dots,\mathop{\widehat{1}}^i,\dots,0)^T.
\]
Thus, we have
\begin{gather*}
\V{c}(\mu^{i-1})=\V{e}_i,\quad \V{c}(F(\tfrac{\mu}{\mu+\gamma})^{p-1})=(\kappa_{p,1},\dots,\kappa_{p,N})^T,
\end{gather*}
where
\[
\textstyle \kappa_{p,q}=\sum_{i=0}^\infty {1-p\choose i}\gamma^{1-p-i}c_{q-p-i}(F)
=\beta^{p-1}\sum_{i=0}^\infty {1-p\choose i}\beta^{i}c_{q-p-i}(F).
\]
The tau function $\tau_n$, as an element of $\CC[\mathcal{O}]$, can be expressed as 
\begin{align*}
\tau_n
&=D^{n-1}\beta^{\frac{(n-1)(n-2)}{2}}
\zet
{
\V{e}_1,\V{e}_2,\dots,\V{e}_{N-n+1},
\V{c}(F),\V{c}(F\tfrac{\mu}{\mu+\gamma}),\dots,\V{c}(F(\tfrac{\mu}{\mu+\gamma})^{n-2})
}\\
&=
D^{n-1}\beta^{\frac{(n-1)(n-2)}{2}}\det(\kappa_{p,N-n+1+q})_{p,q=1}^{n-1}\\
&=
D^{n-1}
\det\left(\textstyle\sum_{i=0}^\infty {1-p\choose i}\beta^{i}c_{N-n+1+q-p-i}(F)\right)_{1\leq p,q\leq n-1}.
\end{align*}
Under natural identification\footnote{
The identification $c_i/c_0\leftrightarrow h_i$ (or $c_i/c_0\leftrightarrow (-1)^ih_i$) appears in \cite{ikeda2017peterson} in order to relate the geometrical information of $Fl_n$ with symmetric polynomials.
The origin of this technique dates back to Fulton's historical work \cite[Part III]{fulton_1996}.
On the other hand, this identification can be understood in terms of the {\it Boson-Fermion correspondence}.
See \cite[Section 6]{ikeda2017peterson}.
}
$c_i/c_0\leftrightarrow h_i$, where $h_i=h_i(x_1,x_2,\dots)$ is the $i^{\mathrm{th}}$ complete symmetric polynomial in infinitely many variables $x_1,x_2,\dots$, and the tau function $\tau_n$ is proportional to 
\[
\det\left(\textstyle\sum_{i=0}^\infty {1-p\choose i}\beta^{i}h_{N-n+1+q-p-i}\right)_{1\leq p,q\leq n-1}.
\]
According to~\cite{Shimozono2003stable,lascoux2014finite}, one finds that this expression exactly coincides with the Jacobi--Trudi type formula for the dual $\beta$-Grothendieck polynomial 
\[
g^{(\beta)}_{R_{n-1}},
\]
where $R_k$ is the Young diagram $R_k=((N-k)^{k})$.

\begin{rem}
The dual $\beta$-Grothendieck polynomial reduces to the Schur polynomial when $\beta=0\colon$
$g^{(0)}_{R_k}=s_{R_k}$.
\end{rem}

\begin{rem}\label{rem:comparement}
By a similar method, one can derive the tau function for 'the other' discrete Toda equation $(\ref{eq:otherform})$.
In fact, it is described as
\begin{equation}\label{eq:anothertau}
\tau_n=\zet{\V{b}(1),\V{b}(\lambda),\dots,\V{b}(\lambda^{N-n}),
\V{b}(G),\V{b}(G\lambda),\dots,\V{b}(G\lambda^{n-2})
},
\end{equation}
where $G\in \mathcal{O}':=\CC[\lambda]/(g(\lambda))$, $g(\lambda)$ is the characteristic polynomial of the Lax matrix of $(\ref{eq:otherform})$ $($see, for example, \cite{hirota1995conserved}$)$ and $\V{b}:\mathcal{O}'\to \CC^N$ is an arbitrary linear isomorphism.
Note that $(\ref{eq:anothertau})$ is invariant under the transformation $\lambda\mapsto \lambda+\gamma$, while $(\ref{eq:tau})$ is not.
This implies that one cannot derive the dual Grothendieck polynomial from this expression.
In fact, the tau function $(\ref{eq:anothertau})$ is naturally related to the determinant
\[
\det(h_{N-n+1+q-p})_{1\leq p,q\leq n-1},
\]
which is the Jacobi--Trudi formula for the Schur polynomial $s_{R_{n-1}}$.
\end{rem}

\section{Solution to the ultradiscrete Toda equation}\label{sec:3}

\subsection{A new ultradiscrete evolution equation}\label{sec:3.1}

Equation (\ref{eq:local}) is equivalent to the following.
\begin{equation}\label{eq:subfreeform}
a_{n+1}^{t+1}=\frac{a_{n+1}^t+b_{n+1}^t}{a_{n}^t+b_{n}^t}a_n^t,\qquad 
b_{n}^{t+1}=\frac{a_{n+1}^t+b_{n+1}^t}{a_{n}^t+b_{n}^t}b_n^t.
\end{equation}
Putting $a_n^t=e^{-\frac{A_n^t}{\ep}}$, $b_n^t=e^{-\frac{B_n^t}{\ep}}$, and taking the limit $\ep\to 0^+$, one derives the ultradiscrete evolution equation given by
\begin{equation}\label{eq:udToda}
\begin{cases}
A_{n+1}^{t+1}=(\min[A_{n+1}^t,B_{n+1}^t]-\min[A_{n}^t,B_{n}^t])+A_n^t,\\
B_{n}^{t+1}=(\min[A_{n+1}^t,B_{n+1}^t]-\min[A_{n}^t,B_{n}^t])+B_n^t,\\
A_0^t=A_{N+1}^t=0,\qquad B_0^t=B_N^t=+\infty.
\end{cases}
\end{equation}
\begin{prop}
Let $U_n^t:=A_n^t-B_n^t+L$.
$(U_0^t=U_N^t:=-\infty.)$
Then, $(\ref{eq:udToda})$ refers to the ultradiscrete Toda equation $(\ref{eq:uToda})$.  
\end{prop}
\begin{proof}
We set $\Delta_n:=\max[0,U_{n}^{t}-L]-\max[0,U_{n-1}^t-L]$.
Equation (\ref{eq:udToda}) is equivalent to\footnote{Note the obvious relation $\max[X,Y]=-\min[-X,-Y]$.}
\begin{gather*}
\Delta_n=A_{n}^{t}-A_{n}^{t+1},\qquad
A_{n+1}^{t+1}-A_n^t=B_n^{t+1}-B_n^t.
\end{gather*}
Therefore, we have
\begin{align*}
&U_{n}^{t+1}-2U_n^t+U_n^{t-1}=A_n^{t+1}-2A_n^{t}+A_{n}^{t-1}-B_n^{t+1}+2B_n^t-B_n^{t-1}\\
&=A_n^{t+1}-2A_n^{t}+A_{n}^{t-1}-(A_{n+1}^{t+1}-A_n^{t})+(A_{n+1}^{t}-A_n^{t-1})\\
&=\Delta_{n+1}-\Delta_n,
\end{align*}
which implies the desired result.
\end{proof}
\begin{rem}
Only when $U_1^t\leq L$, one can recover $A_n^t$ and $B_n^t$ from $U_n^t$ using the formula
\[
A_n^t=\sum_{k=1}^{n-1}(U_{n-k}^{t-1}-U_{n-k}^{t})+C,\qquad 
B_n^t=L-U_n^t+\sum_{k=1}^{n-1}(U_{n-k}^{t-1}-U_{n-k}^{t})+C,
\]
where $C$ is an arbitrary number. 
In fact, we have $A_n^t-A_n^{t+1}=\sum_{k=1}^{n-1}(U_{n-k}^{t+1}-2U_{n-k}^{t}+U_{n-k}^{t-1})=
\sum_{k=1}^{n-1}(\Delta_{n-k+1}-\Delta_{n-k})=\Delta_{n}-\Delta_1=\Delta_n$, where the last equality follows from $U_1^t\leq L\Rightarrow\Delta_1=0$.
\end{rem}

\subsection{Tropical permanent solution}\label{sec:3.2}

In this section, we derive a {\it tropical permanent solution} to (\ref{eq:udToda}).
Let $P=(p_{i,j})_{1\leq i,j\leq N}$ $(p_{i,j}\in \RR\cup \{+\infty\})$ be an $N\times N$ matrix.
The tropical permanent 
$\mathrm{TP}\zet{P}$ is an element of $\RR\cup \{+\infty\}$ defined by the formula
$$
\mathrm{TP}\zet{P}:=\min_{\sigma\in \mathfrak{S}_n}\left[
p_{1,\sigma(1)}+p_{2,\sigma(2)}+\dots+p_{N,\sigma(N)}
\right].
$$
We start with the determinantal solution (\ref{eq:explicit-form-of-tau}) to the discrete Toda equation, wherein we define the tropical permanent $T_n^t$ associated with $\tau_n^t$ by 
\begin{equation}\label{eq:tropicalT}
\begin{aligned}
T_n^t:=\mathrm{TP}
&\left|
\begin{array}{cccc@{\ , \ }ccccc}
(n-1)\Lambda_1 & n\Lambda_1 & \cdots & (N-1)\Lambda_1 
\\ 
(n-1)\Lambda_2 & n\Lambda_2 & \cdots & (N-1)\Lambda_2 
\\ 
\vdots & \vdots & \cdots & \vdots 
\\ 
(n-1)\Lambda_N & n\Lambda_N & \cdots & (N-1)\Lambda_N 
\end{array} 
\right.\\
&\hspace{20pt}
\left.
\begin{array}{ccccc}
F_1+t\Lambda_1 & F_1+(t+1)\Lambda_1 & \cdots & F_1+(t+n-2)\Lambda_1   \\ 
F_2+t\Lambda_2 & F_2+(t+1)\Lambda_2 & \cdots & F_2+(t+n-2)\Lambda_2   \\ 
\vdots & \vdots & \cdots &\vdots \\ 
F_N+t\Lambda_N & F_N+(t+1)\Lambda_N & \cdots & F_N+(t+n-2)\Lambda_N 
\end{array} 
\right|.
\end{aligned}
\end{equation}

\begin{prop}\label{prop:positive_roots}
Let $a_n^t=a_n^t(\ep)$, $b_n^t=b_n^t(\ep)$ be real analytic functions of $\epsilon>0$ with 
$$
\epsilon\ll1 \quad\Rightarrow\quad a_n^t,\ b_n^t>0
$$ 
and 
\[
-\lim\limits_{\ep\to 0^+}\ep\log a_n^t=A_n^t,\qquad
-\lim\limits_{\ep\to 0^+}\ep\log b_n^t=B_n^t.
\]
The following two claims, therefore, hold$\colon$
\begin{enumerate}
\item All eigenvalues of the Lax matrix $L^t$ are distinct and positive.
\[
f(\lambda)=\prod_{i=1}^{N}(\lambda-\lambda_i),\qquad 
\epsilon\ll1 \quad\Rightarrow\quad
0< \lambda_N<\dots<\lambda_1.
\]
\item Under linear isomorphism $(\ref{eq:linear_isom_c})$, the image of $F^t\in \mathcal{O}$ (\S \ref{sec:2.1}) satisfies
\[
\epsilon\ll1 \quad\Rightarrow\quad(-1)^{n}f_n^t>0,\qquad \mbox{where}\quad \V{c}(F^t)=(f_1^t,\dots,f_n^t)^T.
\]
\end{enumerate}
\end{prop}
\begin{proof}  
(1). The distinctness and the positivity of the eigenvalues are direct consequences of the fact that $L^t$ is a {\it totally non-negative and irreducible} matrix (if $\ep\ll 1$).
One may refer the textbook~\cite[Section 5]{pinkus2010totally}.
(2). Comparing the $1^\mathrm{st}$ components on the both sides of (\ref{eq:inverseproblem}), we have $F^t=-\det(\lambda Y^t-X^t)_{1,1}$, where $M_{i,j}$ is the $(i,j)^{\mathrm{th}}$ minor of the matrix $M$.
Using the Cauchy--Binet formula, we derive
\begin{align*}
F^t&=-\det\{(\lambda E -X^t(Y^t)^{-1})\cdot Y^t\}_{1,1}
=-\det(\lambda E -X^t(Y^t)^{-1})_{1,1}\cdot \det Y^t_{1,1}\\
&=-\det(\lambda E-L^{t+1})_{1,1}.
\end{align*}
(The second equality follows from the fact that $Y^t$ is the lower triangle.)
Generally, it is known~\cite[\S 5.3]{pinkus2010totally} that for any totally non-negative and irreducible matrix $A$, the principal minor $p(\lambda)=\det(\lambda E-A)_{1,1}$ satisfies $p(\lambda_1)>0$, $p(\lambda_{2})<0$, $p(\lambda_{3})>0$,\dots, where $0<\lambda_N<\dots<\lambda_1$ are eigenvalues of $A$.
The claim naturally follows from the fact that $L^t$ is totally non-negative and irreducible.
\end{proof}

\begin{prop}\label{prop:Udofdet}
Under the assumptions in proposition \ref{prop:positive_roots},
we write $\Lambda_i=-\lim\limits_{\ep\to 0^+}\ep\log\lambda_i$ and $F_n^t=-\lim\limits_{\ep\to 0^+}\ep\log\zet{f_n^t}$.
Let $\tau_n^t$ be the tau function $(\ref{eq:explicit-form-of-tau})$ and $T_n^t$ be the tropical permanent $(\ref{eq:tropicalT})$.
If $\lambda_1,\dots,\lambda_N$ satisfy the condition\footnote{
If we take $a_n^t,b_n^t$ generically, this condition holds automatically.
}
\begin{equation}\label{eq:generic_condition}
-\lim\limits_{\ep\to 0^+}\ep\log\zet{\lambda_i-\lambda_j}=\min[\Lambda_i,\Lambda_j]
\end{equation}
for any $\lambda_i,\lambda_j$ ($i\neq j$), we have
\[
-\lim\limits_{\ep\to 0^+}\ep \log \zet{\tau_n^t}=T_n^t.
\]
(This implies that the ultradiscretization of $\tau_n^t$ coincides with $T_n^t$.)
\end{prop}

\begin{proof}
From (\ref{eq:expansionoftau}) and the proposition \ref{prop:positive_roots} (2), we have 
\begin{align*} 
 \zet{\tau_n^t}=
\sum_{\genfrac{}{}{0pt}{1}{[N]=\V{j}\sqcup\V{i}}{\sharp\V{i}=n-1}}
\prod_{j\in \V{j}}\lambda_j^{n-1} \cdot
\prod_{i\in \V{i}}|f_i|\lambda_i^{t}\cdot
\prod_{\{j_1<j_2\}\subset \V{j}}(\lambda_{j_1}-\lambda_{j_2})\cdot
\prod_{\{i_1<i_2\}\subset \V{i}}(\lambda_{i_1}-\lambda_{i_2}).  
\end{align*}
Thus, we can derive
\begin{equation*}
  -\lim\limits_{\ep\to 0^+}\ep \log \zet{\tau_n^t}= \min_{\genfrac{}{}{0pt}{1}{[N]=\V{j}\sqcup\V{i}}{\sharp\V{i}=n-1}} \left[ \sum _{k=1}^{N-n+1-k} (N-k)\Lambda _{j_k}+\sum _{l=1}^{n-1} (F_{i_l}+(t+n-1-l)\Lambda _{j_l} \right]
\end{equation*}
using the formula
\begin{equation*} 
  -\lim _{\ep \to 0^+}\ep\log\prod_{\{j_1<j_2\}\subset \V{j}}(\lambda_{j_1}-\lambda_{j_2}) = \sum _{k=1}^{N-n+1}(N-n+1-k)\Lambda _{j_k}.
\end{equation*}
(This formula follows from the condition (\ref{eq:generic_condition}).)
On the other hand, by applying the Laplace expansion for the tropical permanent along the rows between $1^{\mathrm{st}}$ and $(N-n+1)^{\mathrm{th}}$ positions~\cite{nagai2011ultradiscrete}, we get
\begin{align*}  
T_n^t=&
\min_{\genfrac{}{}{0pt}{1}{[N]=\V{j}\sqcup\V{i}}{\sharp\V{i}=n-1}} 
\Biggl[
\mathrm{TP}
\begin{vmatrix}
(n-1)\Lambda_{j_1} & n\Lambda_{j_1} & \cdots & (N-1)\Lambda_{j_1} \\ 
(n-1)\Lambda_{j_2} & n\Lambda_{j_2} & \cdots & (N-1)\Lambda_{j_2} \\ 
\vdots & \vdots & \cdots & \vdots \\ 
(n-1)\Lambda_{j_{N-n+1}} & n\Lambda_{j_{N-n+1}} & \cdots & (N-1)\Lambda_{j_{N-n+1}}  
\end{vmatrix} \\
&\hspace{10pt}
+ \mathrm{TP}
\begin{vmatrix}
 F_{i_1}+t\Lambda_{i_1} & F_{i_1}+(t+1)\Lambda_{i_1} & \cdots & F_{i_1}+(t+n-2)\Lambda_{i_1} \\   
 F_{i_2}+t\Lambda_{i_2} & F_{i_2}+(t+1)\Lambda_{i_2} & \cdots & F_{i_1}+(t+n-2)\Lambda_{i_2} \\   
\vdots & \vdots & \cdots & \vdots \\ 
 F_{i_{n-1}}+t\Lambda_{i_{n-1}} & F_{i_{n-1}}+(t+1)\Lambda_{i_{n-1}} & \cdots & F_{i_{n-1}}+(t+n-2)\Lambda_{i_{n-1}}    
\end{vmatrix}
\Biggr]\\
=& \min_{\genfrac{}{}{0pt}{1}{[N]=\V{j}\sqcup\V{i}}{\sharp\V{i}=n-1}} 
\Biggl[
\sum _{k=1}^{N-n+1} (N-k)\Lambda _{j_k}+\sum _{l=1}^{n-1} (F_{i_l}+(t+n-1-l)\Lambda _{j_l} 
\Biggr].
\end{align*}
Herein, we use the formula 
\begin{align*}  
 \mathrm{TP}
\begin{vmatrix}
\Lambda_{1} & 2\Lambda_{1} & \cdots & N\Lambda_{1} \\ 
\Lambda_{2} & 2\Lambda_{2} & \cdots & N\Lambda_{2} \\ 
\vdots & \vdots & \cdots & \vdots \\ 
\Lambda_{N} & 2\Lambda_{N} & \cdots & N\Lambda_{N}  
\end{vmatrix}   =& \sum _{k=1}^N (N+1-k)\Lambda _k,   
\end{align*}
where $\Lambda _1\leq \Lambda _2\leq \dots \leq \Lambda _N$.  This completes the proof.  
\end{proof}
\begin{rem}
Under the assumptions in the proposition \ref{prop:positive_roots}, 
\begin{equation} 
a_n^t=\frac{|\tau^t_n| |\tau^{{t+1}}_{n+1}|}{|\tau^{t+1}_{n}| |\tau^t_{n+1}|},\qquad
b_n^t=\frac{|\tau^t_n| |\tau^{{t+1}}_{n+2}|}{|\tau^t_{n+1}| |\tau^{{t+1}}_{n+1}|}
\end{equation}
holds.  Thus, the relationships between $T_n^t$ and $A_n^t$, $B_n^t$ are given by  
\begin{equation}\label{eq:solution}
  A_n^t=T^t_n+ T^{{t+1}}_{n+1}-T^{t+1}_{n}-T ^t_{n+1}, \qquad  B_n^t=T^t_n+T^{{t+1}}_{n+2}-T^t_{n+1}- T^{{t+1}}_{n+1}.
\end{equation}
\end{rem}

\subsection{The cellular automaton realization}\label{sec:3.3}

The ultradiscrete evolution equation (\ref{eq:udToda}) can be realized in the form of a cellular automaton as follows.
Consider $N$ cells numbered from $1$ to $N$ (see Figure \ref{fig:ex1}).
At time $t\in \ZZ$, the $n^\mathrm{th}$ cell contains $A_n^t$ 'kickers' and $B_n^t$ 'balls'.
The state at time $t+1$ is obtained by the following rules.
$\bullet$ A kicker kicks out one ball, if it exists, to the cell neighbor to the left.
$\bullet$ A kicker who has no ball to kick out moves to the cell neighbor to the right.
Figure \ref{fig:ex1} illustrates a typical example.
The solution to the ultradiscrete Toda equation (\ref{eq:uToda}) associated with the above scenario is given in Figure \ref{fig:ex2}.

\begin{figure}[htbp]
\begin{center}
\begin{picture}(250,170)
\multiput(0,0)(20,0){14}{\line(0,1){160}}
\multiput(0,0)(0,20){9}{\line(1,0){260}}
\put(-10,165){$n=1$}
\put(27,165){$2$}
\put(47,165){$3$}
\put(67,165){$4$}
\put(85,165){$\cdots$}
\put(205,165){$\cdots$}
\put(225,165){$12$}
\put(245,165){$13$}
\put(-25,145){$t=0$}
\put(-8,126){$1$}
\put(-8,106){$2$}
\put(-8,86){$3$}
\put(-8,66){$\vdots$}
\def\ball#1#2#3#4{%
\dimen0=20pt
\dimen1=20pt
\multiply \dimen0 by #1%
\multiply \dimen1 by #2%
\put(4,150){\hbox{\lower\dimen1 \hbox{\hspace{\dimen0}$#3$}}}%
\put(10,142){\lower\dimen1\hbox{\hspace{\dimen0}$\mathrm{#4}$}}%
}
\def\ballset#1{%
\bgroup
\setcounter{xax}{0}%
\setcounter{yax}{0}%
\def\b##1##2{%
\ball{\thexax}{\theyax}{##1}{##2}%
\addtocounter{xax}{1}%
}%
\def\\ {%
\setcounter{xax}{0}%
\addtocounter{yax}{1}%
}%
#1%
\egroup
}
\ballset{%
\b{.}{i}\b{.}{i}\b{2}{.}\b{1}{i}\b{.}{i}\b{.}{i}\b{1}{.}\b{.}{i}\b{.}{i}\b{.}{i}\b{.}{i}\b{.}{i}\b{.}{\infty}\\
\b{.}{i}\b{.}{i}\b{.}{i}\b{3}{.}\b{.}{i}\b{.}{i}\b{.}{.}\b{1}{i}\b{.}{i}\b{.}{i}\b{.}{i}\b{.}{i}\b{.}{\infty}\\
\b{.}{i}\b{.}{i}\b{.}{i}\b{.}{.}\b{3}{i}\b{.}{i}\b{.}{i}\b{1}{.}\b{.}{i}\b{.}{i}\b{.}{i}\b{.}{i}\b{.}{\infty}\\
\b{.}{i}\b{.}{i}\b{.}{i}\b{.}{i}\b{1}{.}\b{2}{i}\b{.}{i}\b{.}{.}\b{1}{i}\b{.}{i}\b{.}{i}\b{.}{i}\b{.}{\infty}\\
\b{.}{i}\b{.}{i}\b{.}{i}\b{.}{i}\b{.}{i}\b{2}{.}\b{1}{i}\b{.}{i}\b{1}{.}\b{.}{i}\b{.}{i}\b{.}{i}\b{.}{\infty}\\
\b{.}{i}\b{.}{i}\b{.}{i}\b{.}{i}\b{.}{i}\b{.}{i}\b{3}{.}\b{.}{i}\b{.}{.}\b{1}{i}\b{.}{i}\b{.}{i}\b{.}{\infty}\\
\b{.}{i}\b{.}{i}\b{.}{i}\b{.}{i}\b{.}{i}\b{.}{i}\b{.}{.}\b{3}{i}\b{.}{i}\b{1}{.}\b{.}{i}\b{.}{i}\b{.}{\infty}\\
\b{.}{i}\b{.}{i}\b{.}{i}\b{.}{i}\b{.}{i}\b{.}{i}\b{.}{i}\b{1}{.}\b{2}{i}\b{.}{.}\b{1}{i}\b{.}{i}\b{.}{\infty}
}
\end{picture}
\end{center}
\caption{Consider an example of the time evolution (\ref{eq:udToda}), where $N=13$.
Arabic numbers ($1,2,3,\dots$) denote the number of kickers, and Roman numbers ($\mathrm{i},\mathrm{ii},\mathrm{iii},\dots$) denote the number of balls.
A dot denotes $0$.
}\label{fig:ex1}
\end{figure}

\begin{figure}[htbp]
\begin{center}
\verb|..31..2......|\\
\verb|...4..11.....|\\
\verb|...13..2.....|\\
\verb|....22.11....|\\
\verb|.....31.2....|\\
\verb|......4.11...|\\
\verb|......13.2...|\\
\verb|.......2211..|
\end{center}
\caption{Solution to the ultradiscrete Toda equation associated with the example in Figure \ref{fig:ex1}, where $L=1$.
Two solitons can be observed to be proceeding from left to right.
}\label{fig:ex2}
\end{figure}

Another example is given in Figure \ref{fig:ex3}, where a traveling soliton and a static soliton~\cite{hirota2009new} interact.
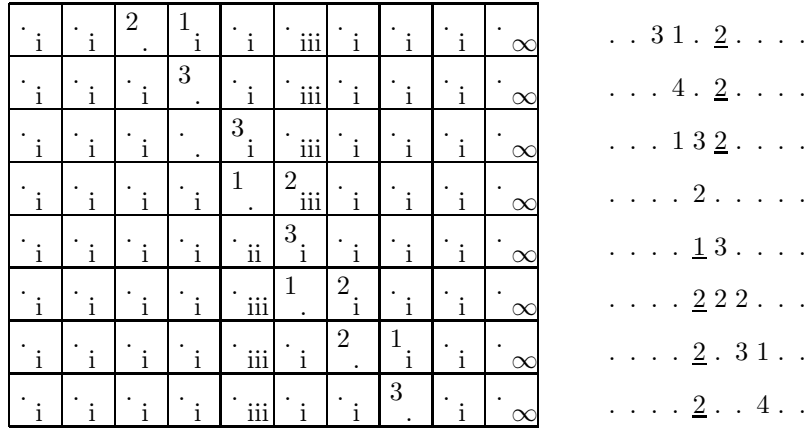
\begin{figure}[htbp]
\begin{center}
\begin{picture}(200,170)
\multiput(0,0)(20,0){11}{\line(0,1){160}}
\multiput(0,0)(0,20){9}{\line(1,0){200}}
\def\ball#1#2#3#4{%
\dimen0=20pt
\dimen1=20pt
\multiply \dimen0 by #1%
\multiply \dimen1 by #2%
\put(4,150){\hbox{\lower\dimen1 \hbox{\hspace{\dimen0}$#3$}}}%
\put(10,142){\lower\dimen1\hbox{\hspace{\dimen0}$\mathrm{#4}$}}%
}
\def\ballset#1{%
\bgroup%
\setcounter{xax}{0}%
\setcounter{yax}{0}%
\def\b##1##2{%
\ball{\thexax}{\theyax}{##1}{##2}%
\addtocounter{xax}{1}%
}%
\def\\ {%
\setcounter{xax}{0}%
\addtocounter{yax}{1}%
}%
#1%
\egroup%
}
\ballset{%
\b{.}{i}\b{.}{i}\b{2}{.}\b{1}{i}\b{.}{i}\b{.}{iii}\b{.}{i}\b{.}{i}\b{.}{i}\b{.}{\infty}\\
\b{.}{i}\b{.}{i}\b{.}{i}\b{3}{.}\b{.}{i}\b{.}{iii}\b{.}{i}\b{.}{i}\b{.}{i}\b{.}{\infty}\\
\b{.}{i}\b{.}{i}\b{.}{i}\b{.}{.}\b{3}{i}\b{.}{iii}\b{.}{i}\b{.}{i}\b{.}{i}\b{.}{\infty}\\
\b{.}{i}\b{.}{i}\b{.}{i}\b{.}{i}\b{1}{.}\b{2}{iii}\b{.}{i}\b{.}{i}\b{.}{i}\b{.}{\infty}\\
\b{.}{i}\b{.}{i}\b{.}{i}\b{.}{i}\b{.}{ii}\b{3}{i}\b{.}{i}\b{.}{i}\b{.}{i}\b{.}{\infty}\\
\b{.}{i}\b{.}{i}\b{.}{i}\b{.}{i}\b{.}{iii}\b{1}{.}\b{2}{i}\b{.}{i}\b{.}{i}\b{.}{\infty}\\
\b{.}{i}\b{.}{i}\b{.}{i}\b{.}{i}\b{.}{iii}\b{.}{i}\b{2}{.}\b{1}{i}\b{.}{i}\b{.}{\infty}\\
\b{.}{i}\b{.}{i}\b{.}{i}\b{.}{i}\b{.}{iii}\b{.}{i}\b{.}{i}\b{3}{.}\b{.}{i}\b{.}{\infty}
}
\end{picture}
\hspace{20pt}
\begin{picture}(100,170)
\def\bal#1#2#3{%
\dimen0=8pt
\dimen1=20pt
\multiply \dimen0 by #1%
\multiply \dimen1 by #2%
\put(0,145){\hbox{\lower\dimen1 \hbox{\hspace{\dimen0}$#3$}}}%
}
\def\balset#1{%
\bgroup
\setcounter{xax}{0}%
\setcounter{yax}{0}%
\def\d##1{%
\bal{\thexax}{\theyax}{##1}%
\addtocounter{xax}{1}%
}%
\def\\ {%
\setcounter{xax}{0}%
\addtocounter{yax}{1}%
}%
\def\u##1{%
\underline{##1}%
}%
#1%
\egroup
}
\balset{%
\d{.}\d{.}\d{3}\d{1}\d{.}\d{\u{2}}\d{.}\d{.}\d{.}\d{.}\\
\d{.}\d{.}\d{.}\d{4}\d{.}\d{\u{2}}\d{.}\d{.}\d{.}\d{.}\\
\d{.}\d{.}\d{.}\d{1}\d{3}\d{\u{2}}\d{.}\d{.}\d{.}\d{.}\\
\d{.}\d{.}\d{.}\d{.}\d{2}\d{.}\d{.}\d{.}\d{.}\d{.}\\
\d{.}\d{.}\d{.}\d{.}\d{\u{1}}\d{3}\d{.}\d{.}\d{.}\d{.}\\
\d{.}\d{.}\d{.}\d{.}\d{\u{2}}\d{2}\d{2}\d{.}\d{.}\d{.}\\
\d{.}\d{.}\d{.}\d{.}\d{\u{2}}\d{.}\d{3}\d{1}\d{.}\d{.}\\
\d{.}\d{.}\d{.}\d{.}\d{\u{2}}\d{.}\d{.}\d{4}\d{.}\d{.}
}
\end{picture}
\end{center}
\caption{Collision between a traveling soliton and a static soliton. ($L=1$).
An underlined number denotes a negative integer ($\underline{2}=-2$).
}\label{fig:ex3}
\end{figure}

\section{Example and Concluding Remarks}\label{sec:4}

\subsection{Example}

As seen above, one can construct the　tropical permanent solution to the ultradiscrete Toda equation for {\it any} initial state. 
The method we use here can be referred to as the {\it tropical inverse scattering method}.
In this section, we demonstrate the method to construct the tropical permanent solution through an example.

Let us consider the initial state ($N=8$, $L=1$)
\begin{gather*}
(A_1^0,A_2^0,\dots,A_8^0)=(1,1,1,1,0,0,0,0),\quad
(B_1^0,B_2^0,\dots,B_7^0)=(2,2,2,4,1,1,1).
\end{gather*}
Let 
$q:=e^{-\frac{1}{\ep}}$.
Set the initial values
\begin{gather*}
(a_1^0,a_2^0,\dots,a_8^0)=(q,2q,3q,4q,1,2,3,4),\quad
(b_1^0,b_2^0,\dots,b_7^0)=(q^2,q^2,q^2,q^4,q,q,q)
\end{gather*}
for the discrete Toda equation.
The characteristic polynomial of the Lax matrix $L^0$ is calculated as 
$f(\lambda)=\sum_{i=0}^8(-1)^iI_i\lambda^{8-i}$, where 
\begin{gather*}
I_0=1,\quad
I_1=10+13q+3q^2+q^4,\quad 
I_2=35+115q+96q^2+24q^3+10q^4+8q^5+2q^6,\\
I_3=50+368q+605q^2+360q^3+102q^4+63q^5+41q^6+8q^7,
\cdots,I_8=576q^4.
\end{gather*}
For sufficiently small $q>0$, the roots $0<\lambda_8<\dots<\lambda_1$ of $f(\lambda)$ could be expanded as 
\begin{gather*}
\textstyle
\lambda_1=4+4q-4q^2+\cdots,\quad
\textstyle
\lambda_2=3+{{9q^2}\over{2}}-{{81q^3}\over{4}}+\cdots,\\
\textstyle
\lambda_3=2-2q^2+6q^3+\cdots,\quad
\textstyle
\lambda_4=1-q+{{3q^2}\over{2}}+\cdots,\\
\textstyle
\lambda_5=4q+4q^2-4q^3+\cdots,\quad
\lambda_6=3q+\frac{9}{2}q^3+\cdots,\quad
\lambda_7=2q+\cdots,\quad 
\lambda_8=q+\cdots.
\end{gather*}
(Higher-order terms have been omitted owing to limitations of space. 
If one needs to execute all calculations below, many higher-order terms would be required. 
For example, $\lambda_1$ must be calculated up to the term of $q^{13}$, whose coefficient is ${{27560920906072627}\over{11337408}}$.)
For each $\lambda_i$, $f_i=-\det(\lambda_i Y-X)_{1,1}$ is calculated as follows$\colon$
\begin{gather*}
\textstyle
f_1=-\frac{32}{3}q^{13}+\cdots,\quad 
f_2=\frac{9}{2}q^{12}+\cdots,\quad
f_3=-4q^{11}+\cdots,\quad
f_4=6q^{10}+\cdots,\\
f_5=-256q^{6}+\cdots,\quad
f_6=108q^5+\cdots,\quad
f_7=-96q^4+\cdots,\quad
f_8=144q^3+\cdots.
\end{gather*}
Therefore, we have 
\[
(\Lambda_1,\dots,\Lambda_8)=(0,0,0,0,1,1,1,1),\qquad (F_1,\dots,F_8)=(13,12,11,10,6,5,4,3).
\]
Substituting these datum values to the tropical permanent $T_n^t$ (\ref{eq:tropicalT}), we obtain:
\def\Trop{}
\begin{gather*}
T_1^t=6,\quad
T_2^t=\min[20,t+9],\quad
T_3^t=\min[35,t+22,2t+13],\\
T_4^t=\min[51,t+36,2t+25,3t+18],\quad
T_5^t=\min[68,t+51,2t+38,3t+29,4t+24],\\
T_6^t=\min[t+67,2t+52,3t+41,4t+34],\quad
T_7^t=\min[2t+67,3t+54,4t+45],\\
T_8^t=\min[3t+68,4t+57],\quad
T_9^t=4t+70.
\end{gather*}
Further, substituting them into the formulas
$A_n^t=T_n^t+T_{n+1}^{t+1}-T_n^{t+1}-T_{n+1}^t$,
$B_n^t=T_n^t+T_{n+2}^{t+1}-T_{n+1}^{t+1}-T_{n+1}^t$, and
$U_n^t=L+A_n^t-B_n^t=1+A_n^t-B_n^t$,
we obtain the solution to the ultradiscrete Toda equation (\ref{eq:uToda}), as represented in Figure \ref{fig:4}.

\begin{figure}[htbp]
\begin{center}
\begin{picture}(50,150)
\def\bal#1#2#3{%
\dimen0=6pt
\dimen1=14pt
\multiply \dimen0 by #1%
\multiply \dimen1 by #2%
\put(0,145){\hbox{\lower\dimen1 \hbox{\hspace{\dimen0}$#3$}}}%
}
\def\balset#1{%
\bgroup
\setcounter{xax}{0}%
\setcounter{yax}{0}%
\def\d##1{%
\bal{\thexax}{\theyax}{##1}%
\addtocounter{xax}{1}%
}%
\def\\ {%
\setcounter{xax}{0}%
\addtocounter{yax}{1}%
}%
\def\u##1{%
\underline{##1}%
}%
#1%
\egroup
}
\balset{%
\hbox to 0pt{$t=0$ -- $9$}
\d{.}\d{.}\d{.}\d{\u{2}}\d{.}\d{.}\d{.}\\
\d{.}\d{.}\d{.}\d{\u{1}}\d{.}\d{.}\d{.}\\
\d{.}\d{.}\d{.}\d{.}\d{.}\d{.}\d{.}\\
\d{.}\d{.}\d{.}\d{1}\d{.}\d{.}\d{.}\\
\d{.}\d{.}\d{.}\d{2}\d{.}\d{.}\d{.}\\
\d{.}\d{.}\d{1}\d{1}\d{1}\d{.}\d{.}\\
\d{.}\d{.}\d{2}\d{.}\d{2}\d{.}\d{.}\\
\d{.}\d{1}\d{1}\d{1}\d{1}\d{1}\d{.}\\
\d{.}\d{2}\d{.}\d{2}\d{.}\d{2}\d{.}\\
\d{1}\d{1}\d{1}\d{1}\d{1}\d{1}\d{1}
}
\end{picture}
\hspace{60pt}
\begin{picture}(50,150)
\def\bal#1#2#3{%
\dimen0=6pt
\dimen1=14pt
\multiply \dimen0 by #1%
\multiply \dimen1 by #2%
\put(0,145){\hbox{\lower\dimen1 \hbox{\hspace{\dimen0}$#3$}}}%
}
\def\balset#1{%
\bgroup
\setcounter{xax}{0}%
\setcounter{yax}{0}%
\def\d##1{%
\bal{\thexax}{\theyax}{##1}%
\addtocounter{xax}{1}%
}%
\def\\ {%
\setcounter{xax}{0}%
\addtocounter{yax}{1}%
}%
\def\u##1{%
\underline{##1}%
}%
#1%
\egroup
}
\balset{%
\hbox to 0pt{$t=10$ -- $19$}
\d{2}\d{.}\d{2}\d{.}\d{2}\d{.}\d{2}\\
\d{1}\d{1}\d{1}\d{1}\d{1}\d{1}\d{1}\\
\d{.}\d{2}\d{.}\d{2}\d{.}\d{2}\d{.}\\
\d{.}\d{1}\d{1}\d{1}\d{1}\d{1}\d{.}\\
\d{.}\d{.}\d{2}\d{.}\d{2}\d{.}\d{.}\\
\d{.}\d{.}\d{1}\d{1}\d{1}\d{.}\d{.}\\
\d{.}\d{.}\d{.}\d{2}\d{.}\d{.}\d{.}\\
\d{.}\d{.}\d{.}\d{1}\d{.}\d{.}\d{.}\\
\d{.}\d{.}\d{.}\d{.}\d{.}\d{.}\d{.}\\
\d{.}\d{.}\d{.}\d{\u{1}}\d{.}\d{.}\d{.}
}
\end{picture}
\end{center}
\caption{Example of the solution to the ultradiscrete Toda equation. ($N=8,L=1$).}
\label{fig:4}
\end{figure}

\subsection{Concluding remarks}

The primary objective of this paper has been to understand the algebraic structure and positivity of the discrete Toda equation with boundary conditions.
The solution itself is constructed in a straightforward manner.
Under certain natural identifications, it is possible to obtain a family of special solutions, which correspond to dual $\beta$-Grothendieck polynomials, which is the $K$-theoretic analogue of Schur polynomials.
One could infer that the discrete Toda equation is of the 'Grothendieck' polynomial type ($\simeq  K$-theoretical), while the other discrete Toda equation is of the 'Schur' polynomial type.
This result can be expected to clarify deeper structures of the two discrete Toda equations.

The ultradiscrete analogues of the Toda equation have also been studied.
It is proved that the ultradiscrete Toda equation reduces to a new evolution equation, which is the ultradiscretization of the Lax formula.
The tropical permanent solutions have also been given. 
Our result shows the correspondence between the determinant solution to the discrete Toda equation and the tropical permanent solution to the ultradiscrete Toda equation.  Moreover, the cellular automaton realization of the Toda equation is proposed.
The proposed method has the inherent advantage of being applicable to arbitrary initial values.
Especially, we have generalized Hirota's formula for static-solitons.

\section*{Acknowledgment}

One of the authors (S.I.) is partially supported by KAKENHI (26800062).
We would like to thank Editage (www.editage.jp) for English language editing.

\bibliographystyle{abbrv-iwao}
\bibliography{IN-ref}

\end{document}